\documentclass[11pt]{article}

\usepackage[T1]{fontenc}
\usepackage{lmodern}
\usepackage{fullpage}
\usepackage{authblk}
\usepackage{amssymb}
\usepackage[english]{babel}
\usepackage[utf8]{inputenc}
\usepackage{enumerate}
\usepackage{cancel}
\usepackage{xspace}
\usepackage{bm}

\usepackage{microtype}
\usepackage{amsthm}
\usepackage{amsmath}
\usepackage{amsfonts}
\usepackage{graphicx}
\usepackage[colorinlistoftodos]{todonotes}
\usepackage{mathtools}

\usepackage{color}
\definecolor{darkgreen}{rgb}{0,0.5,0}
\definecolor{darkblue}{rgb}{0,0,0.8}
\definecolor{darkred}{rgb}{0.8,0,0}

\usepackage{makecell}
\usepackage{tikz}
\usetikzlibrary{shapes,arrows,positioning}

 \usepackage{hyperref}
 \hypersetup{
   unicode=false,          
   colorlinks=true,        
   linkcolor=darkred,          
   citecolor=darkblue,        
   filecolor=magenta,      
   urlcolor=blue           
 }

\newtheorem{definition}{Definition}[section]
\newtheorem{lemma}[definition]{Lemma}
\newtheorem{theorem}[definition]{Theorem}

\newtheorem{algo}[definition]{Algorithm}

\graphicspath{{./figures/}}

\newcommand{\bigo}{\mathcal{O}}
\newcommand{\ham}{\mathrm{Ham}}

\DeclarePairedDelimiter\floor{\lfloor}{\rfloor}
\newcommand{\etconst}{15555}
\newcommand{\ExpR}[1]{\mathbb{E}_{r \in [1,9)}\Big[#1\Big]}
\newcommand{\uglyconst}{1152  + 192 (\ln 2))}
\newcommand{\mergedconst}{2593} 
\newcommand{\given}{\Big\vert}
\newcommand{\PrR}[1]{\Pr_{r \in [1,9)}\Big[#1\Big]}
\newcommand{\indicator}[1]{\ensuremath{\mathbf{1}\!\left[#1\right]}}

\newcommand\defeq{\stackrel{\mathclap{\normalfont{\mbox{\text{\tiny{def}}}}}}{=}}
\newcommand{\snote}[1]{&\qquad\text{(#1)}}
\newcommand{\bnote}[1]{&\qquad\parbox[t]{6cm}{(#1)}}

\renewcommand{\paragraph}{\subparagraph}

 \title{\bf Approximating Approximate Pattern Matching}

 \author[1]{\Large Jan Studen\'y}
 \author[2]{\Large Przemys\l{}aw~Uzna\'nski}
 \affil[1]{ETH Zürich, Switzerland}
 \affil[2]{University of Wrocław, Poland}
 \date{}

\begin{document}
\maketitle

 \thispagestyle{empty}


\begin{abstract}
Given a text $T$ of length $n$ and a pattern $P$ of length $m$, the approximate pattern matching problem asks for computation of a particular \emph{distance} function between $P$ and every $m$-substring of $T$.
We consider a $(1\pm\varepsilon)$ multiplicative approximation variant of this problem, for $\ell_p$ distance function.
In this paper, we describe two $(1+\varepsilon)$-approximate algorithms with a runtime of $\widetilde\bigo(\frac{n}{\varepsilon})$ for all (constant) non-negative values of $p$.
For constant $p \ge 1$ we show a deterministic $(1+\varepsilon)$-approximation algorithm. Previously, such run time was known only for the case of $\ell_1$ distance, by Gawrychowski and Uznański [ICALP 2018] and only with a randomized algorithm.
For constant $0 \le p \le 1$ we show a randomized algorithm for the $\ell_p$, thereby providing a smooth tradeoff between algorithms of Kopelowitz and Porat [FOCS~2015, SOSA~2018] for Hamming distance (case of $p=0$) and of Gawrychowski and Uznański for $\ell_1$ distance.
\end{abstract}

\section{Introduction}
Pattern matching is one of the core problems in text processing algorithms. Given a text $T$ of length $n$ and a pattern $P$ of length $m$, $m \leq n$, both over an alphabet $\Sigma$, one searches for occurrences of $P$ in $T$ as a substring. A generalization of a pattern matching is to find substrings of $T$ that are similar to $P$, where we consider a particular string distance and ask for all $m$-substrings of $T$ where the distance to $P$ does not exceed a given threshold, or simply report the distance from $P$ to every $m$-substring of $T$. Typical distance functions considered are Hamming distance, $\ell_1$ distance, or in general $\ell_p$ distances for some constant $p$, assuming input is over a numerical, e.g. integer, alphabet.

For reporting all Hamming distances, Abrahamson \cite{Abrahamson87} described an algorithm with the complexity of $\bigo(n\sqrt{m \log m})$. Using a similar approach, the same complexity was obtained in \cite{lipsky2003eficient} and later in conference works \cite{Amir2005,DBLP:conf/cpm/CliffordCI05} for reporting all $\ell_1$ distances. It is a major open problem whether near-linear time algorithm, or even $\bigo(n^{3/2-\varepsilon})$ time algorithm, is possible for such problems. A conditional lower bound~\cite{Clifford} was shown, via a reduction from matrix multiplication. This means that existence of combinatorial algorithm with runtime $\bigo(n^{3/2-\varepsilon})$ solving the problem for Hamming distances implies  combinatorial algorithms for boolean matrix multiplication with $\bigo(n^{3-\delta}) $ runtime, which existence is unlikely. If one is uncomfortable with poorly defined notion of \emph{combinatorial} algorithms, one can apply the reduction to obtain a lowerbound of $\Omega(n^{\omega/2})$ for Hamming distances pattern matching, where $2 \le \omega < 2.373$ is a matrix multiplication exponent.\footnote{Although the issue is that we do not even know whether $\omega>2$ or not.} Later, the complexity of pattern matching under Hamming distance and under $\ell_1$ distance was proven to be identical (up to polylogarithmic terms) \cite{DBLP:conf/icalp/0001LU18,DBLP:journals/ipl/LipskyP08a}.

The mentioned hardness results serve as a motivation for considering relaxation of the problems, with  $(1+\varepsilon)$ multiplicative approximation being the obvious candidate. For Hamming distance, Karloff \cite{Karloff93} was the first to propose an efficient approximation algorithm with a run time of $\bigo(\frac{n}{\varepsilon^2}\log^3m)$. The $\frac{1}{\varepsilon^2}$  dependency was believed to be inherent, as is the case for e.g. space complexity of sketching of Hamming distance, cf. \cite{DBLP:conf/soda/Woodruff04,DBLP:journals/toc/JayramKS08,DBLP:journals/siamcomp/ChakrabartiR12}. However, for approximate pattern matching that was refuted by Kopelowitz and Porat \cite{DBLP:conf/focs/KopelowitzP15,DBLP:conf/soda/KopelowitzP18}, by providing randomized algorithms with complexity $\bigo(\frac{n}{\varepsilon} \log n \log m \log \frac{1}{\varepsilon} \log |\Sigma|)$ and $\bigo(\frac{n}{\varepsilon} \log n \log m)$ respectively. Moving to $\ell_1$ distance, Lipsky and Porat \cite{DBLP:journals/algorithmica/LipskyP11} gave a deterministic algorithm with a run time of  $\bigo(\frac{n}{\varepsilon^2} \log m \log U)$, while later Gawrychowski and Uznański \cite{DBLP:conf/icalp/GawrychowskiU18} have improved the complexity to a (randomized) $\bigo(\frac{n}{\varepsilon} \log^2 n \log m \log U)$, where $U$ is the maximal integer value on the input. Additionally, we refer the reader to the line of work on other relaxations on exact the distance reporting  \cite{DBLP:journals/jal/AmirLP04,CliffordFPSS16,DBLP:conf/icalp/GawrychowskiU18,Amir2005}.

A folklore result  (c.f. \cite{DBLP:journals/algorithmica/LipskyP11}) states that the randomized algorithm with a run time of  $\widetilde\bigo(\frac{n}{\varepsilon^2})$ is in fact possible for any $\ell_p$ distance, $0<p\le2$, with use of $p$-stable distributions and convolution.\footnote{We use $\widetilde\bigo$ notation to hide factors polylogarithmic in $n, m, |\Sigma|, U$ and $\varepsilon^{-1}$.} Such distributions exist only when $p\le2$, which puts a limit on this approach. See \cite{nolan2003stable} for wider discussion on $p$-stable distributions.
Porat and Efremenko \cite{DBLP:conf/soda/PoratE08} have shown how to approximate general distance functions between pattern and text in time $\bigo(\frac{n}{\varepsilon^2} \log^2 m \log^3 |\Sigma| \log B_d)$, where $B_d$ is upperbound on distance between two characters in $\Sigma$. Their solution does not immediately translates to $\ell_p$ distances, since it allows only for score functions of form $\sum_{j} d(t_{i+j}, p_j)$ where $d$ is arbitrary metric over $\Sigma$. Authors state that their techniques generalize to computation of $\ell_2$ distances, but the dependency $\varepsilon^{-2}$ in their approach is unavoidable.
 \cite{DBLP:journals/algorithmica/LipskyP11} observe that $\ell_2$ pattern matching can be in fact computed in $\bigo(n \log m)$ time, by reducing it to a single convolution computation. This case and analogously case of $p=4,6,\ldots$ are the only ones where fast and exact algorithm is known.

We want to point that for $\ell_\infty$ pattern matching there is an approximation algorithm of complexity $\bigo(\frac{n}{\varepsilon} \log m \log U)$ by Lipsky and Porat \cite{DBLP:journals/algorithmica/LipskyP11}. Moving past pattern matching, we want to point that in a closely related problem of computing $(\min,+)$-convolution there exists $\bigo(\frac{n}{\varepsilon} \log \frac{n}{\varepsilon} \log U)$ time algorithm computing $(1+\varepsilon)$ approximation, cf. Mucha et al. \cite{DBLP:conf/soda/MuchaW019}.

 Two questions follow naturally. First, is there a $\widetilde\bigo(\frac{n}{\textrm{poly}(\varepsilon)})$ algorithm for $\ell_p$ norms pattern matching when $p > 2$? Second, is there anything special to $p=0$ and $p=1$ cases that allows for faster algorithms, or can we extend their complexities to other $\ell_p$ norms? To motivate further those questions, observe that in the regime of maintaining $\ell_p$ sketches in the \emph{turnstile} streaming model (sequence of updates to vector coordinates), one needs small space of $\Theta(\log n)$ bits when $p \le 2$ (cf. \cite{DBLP:conf/soda/KaneNW10}), while when $p > 2$ one needs large space of $\Theta(n^{1-2/p} \log n)$ bits (cf. \cite{DBLP:conf/icalp/Ganguly15,DBLP:conf/approx/LiW13}) meaning there is a sharp transition in problem complexity at $p=2$. Similar phenomenon of transition at $p=2$ is observed for $p$-stable distributions, and one could expect such transition to happen in the pattern matching regime as well.

In this work we show that for any \emph{constant} $p \ge 0$ there is an algorithm of complexity $\widetilde\bigo(\frac{n}{\varepsilon})$, replicating the phenomenon of linear dependency on $\varepsilon^{-1}$ from Hamming distance and $\ell_1$ distance to all $\ell_p$ norms. Additionally this provides evidence that no transition at $p=2$ happens, and so far to our understanding cases of $p>2$ and $p<2$ are of similar hardness.

\subsection{Definitions and preliminaries.}
\paragraph{Model.}
In most general setting, our inputs are strings taken from arbitrary alphabet $\Sigma$. We use this notation only when structure of alphabet is irrelevant for the problem (e.g. Hamming distances). However, when considering $\ell_p$ distances we focus our attention over an integer alphabet $[U] \defeq \{0,1,...,U-1\}$ for some $U$. One can usually assume that $U = \textrm{poly}(n)$, and then $\log U$ term can be safely hidden in the $\widetilde\bigo$ notation, however we provide the dependency explicitly in Theorem statements. Even without such assumption, we can assume standard word RAM model, in which arithmetic operations on words of size $\log U$ take constant time. Otherwise the complexities have an additional $\log U$ factor. We also denote $u = \log U$. While we restrict input integer values, we allow intermediate computation and output to consist of floating point numbers having $u$ bits of precision.
\paragraph{Distance between strings.}
Let $X = x_1 x_2 \ldots x_n$ and $Y = y_1 y_2 \ldots y_n$ be two strings. For any $p > 0$, we define their $\ell_p$ distance as
$$\ell_p(X,Y) = \left(\sum_i |x_i - y_i|^p \right)^{1/p}.$$
Particularly, $\ell_1$ distance is known as \emph{Manhattan distance}, and $\ell_2$ distance is known as \emph{Euclidean distance}. Observe that the $p$-th power of $\ell_p$ distance has particularly simpler form of $\ell_p(X,Y)^p = \sum_i |x_i - y_i|^p$.

The \emph{Hamming distance} between two strings is defined as $$\ham(X,Y) =  | \{ i : x_i \not= y_i \}|.$$
Adopting the convention that $0^0 = 0$ and $x^0 = 1$ for $x \not= 0$, we observe that $(\ell_p)^p$ approaches Hamming distance as $p \to 0$. Thus Hamming distance is usually denoted as $\ell_0$ (although $(\ell_0)^0$ is more precise notation).

\paragraph{Text-to-pattern distance.}
For text $T = t_1t_2\ldots t_n$ and pattern $P = p_1p_2\ldots p_m$, the text-to-pattern distance is defined as an array $S$ such that, for every $i$, $S[i] = d(T[i+1\ ..\ i+m],P)$ for particular distance function $d$. Thus, for $\ell_p$ distance $S[i] = \left(\sum_{j=1}^m |t_{i+j}-p_j|^p\right)^{1/p}$, while for Hamming distance $S[i] = |\{ j \in \{1, \ldots, m\} : t_{i+j} \not= p_j \}|$. Then $(1+\varepsilon)$-approximate distance is defined as an array $S_{\varepsilon}$ such that, for every $i$, $(1-\varepsilon) \cdot S[i] \le S_{\varepsilon}[i] \le (1+\varepsilon) \cdot S[i]$.

\paragraph{Rounding and arithmetic operations.}
For any value $x$, we denote by $x^{(i)} = \lfloor x/2^{i} \rfloor \cdot 2^i$ the value with $i$ y bits rounded. However, with a little stretch of notation, we do not limit value of $i$ to be positive.
We denote by $\|r\|_c$ the norm \emph{modulo} $c$, that is
$\|r\|_c = \min( r \bmod c, c - (r \bmod c))$.

\subsection{Our results.}
In this paper we answer favorably both questions by providing relevant algorithms. First, we show how to extend the deterministic $\ell_1$ distances algorithm into $\ell_p$ distances, when $p \ge 1$.

\begin{theorem}
\label{th:large_p}
For any $p\ge1$ there is a \emph{deterministic} algorithm computing $(1+\varepsilon)$ approximation to  pattern matching under $\ell_p$ distances in time $\bigo(\frac{n}{\varepsilon} \log m \log U)$ (assuming $\varepsilon \le 1/p$).
\end{theorem}

We then move to the case of $\ell_p$ distances when $p < 1$. We show that it is possible to construct a randomized algorithm with the desired complexity.

\begin{theorem}
\label{th:small_p}
For $0 < p < 1$, there is a \emph{randomized} algorithm computing $(1+\varepsilon)$ approximation to pattern matching under $\ell_p$ distances in time $\bigo(p^{-1} \varepsilon^{-1}n \log m \log^2 U \log n)$. The algorithm is correct with high probability.\footnote{Probability at least $1-1/n^c$ for arbitrarily large constant $c$.}
\end{theorem}

Finally, combining with existing $\ell_0$ algorithm from \cite{DBLP:conf/soda/KopelowitzP18} we obtain as a corollary that for \emph{constant} $p \ge 0$ approximation of pattern matching under $\ell_p$ distances can be computed in $\widetilde\bigo(\frac{n}{\varepsilon})$ time.


\section{Approximation of $\ell_p$ distances}
We start by showing how convolution finds its use  in counting versions of pattern matching, either exact or approximation algorithms. Consider the case of pattern matching under $\ell_2$ distances. Observe that we are looking for $S$ such that $S[i]^2 = \sum_{j-k=i} (t_j - p_k)^2 = \sum_j t_j^2 + \sum_k p_k^2 - 2 \sum_{j-k=i} t_j p_k$. The last  term is just a convolution of vectors in disguise and is equivalent to computing convolution of $T$ and reverse ordered $P$. Such approach can be applied to solving exact pattern matching via convolution (observing that $\ell_2$ distance is 0 iff there is an exact match).

We follow with a technique for computing exact text-to-pattern distance, for arbitrary distance functions, introduced by \cite{DBLP:journals/algorithmica/LipskyP11}, which is a generalization of a technique used in \cite{Fischer:1974:SOP:889566}. We provide a short proof for completeness.
\begin{theorem}[\cite{DBLP:journals/algorithmica/LipskyP11}]
\label{th:lipskyporat}
Text-to-pattern distance where strings are over \emph{arbitrary} alphabet $\Sigma$ can be computed \emph{exactly} in time $\bigo(|\Sigma| \cdot n \log m)$.
\end{theorem}
\begin{proof}
For every letter $c \in \Sigma$, construct a new text $T^c$ by setting $T^c[i] = 1$ if $t_i=c$ and $T^c[i] = 0$ otherwise. A new pattern $P^c$ is constructed by setting $P^c[i] = d(c,p_i)$. Since $d(t_{i+j},p_j) = \sum_{c \in \Sigma} T^c[i+j] \cdot P^c[j]$, it is enough to invoke $|\Sigma|$ times convolution.
\end{proof}

Theorem~\ref{th:lipskyporat} allows us to compute text-to-pattern distance exactly, but the time complexity $\bigo(|\Sigma|n \log m)$ is prohibitive for large alphabets (when $|\Sigma| = poly(n)$). However, it is enough to reduce the size of alphabet used in the problem (at the cost of reduced precision) to reach desired time complexity. While this might be hard, we proceed as follows: we decompose our weight function into a sum of components, each of which is approximated by a corresponding function on a reduced alphabet.

We say that a function $d$ is \emph{effectively over smaller alphabet $\Sigma'$} if it is represented as $d(x,y) = d'(\iota_1(x), \iota_2(y))$ for some  $\iota_1,\iota_2 : \Sigma \to \Sigma'$ and $d'$. It follows from Theorem~\ref{th:lipskyporat} that text-to-pattern under distance $d$  can be computed in time $\widetilde\bigo(|\Sigma'| n)$ (ignoring the cost of computing $\iota_1$ and $\iota_2$).

\paragraph{Decomposition.}
Let $D(x,y) = |x-y|^p$ be a function corresponding to $(\ell_p)^p$ distance, that is $\ell_p(X,Y)^p = \sum_i D(x_i,y_i)$.
Our goal is to decompose $D(x,y) = \sum_i \alpha_i(x,y)$ into small (polylogarithmic) number of functions, such that each $\alpha_i(x,y)$ is approximated by $\beta_i(x,y)$ that is effectively over alphabet of $\bigo(\frac{1}{\varepsilon})$ size (up to polylogarithmic factors). Now we can use Theorem~\ref{th:lipskyporat} to compute contribution of each $\beta_i$. We then have that $G(x,y) = \sum_i \beta_i(x,y)$ approximates $F$, and text-to-pattern distance under $G$ can be computed in the desired $\widetilde\bigo(\frac{n}{\varepsilon})$ time. We present such decomposition, useful immediately in case of $p \ge 1$ and as we see in section \ref{sec:smallp} with a little bit of effort as well in case when $0 < p \le 1$.

\paragraph{Useful estimations.}
We use following estimations in our proofs.
For $p \ge 1$
\begin{align}
\label{eq:bound1}
&(1-\varepsilon)^p \ge 1-p\varepsilon,&\text{for } 0 \le \varepsilon \le 1,\\
\label{eq:bound1a}
&(1+\varepsilon)^p \ge 1+p\varepsilon,&\text{for } 0 \le \varepsilon,\\
\label{eq:bound4a}
&(1-\varepsilon)^p \le 1 - p\varepsilon(1 - 1/e),&\text{for } 0 \le \varepsilon \le 1/p,\\
\label{eq:bound_derived_p1inf}
&a^p-(a-b)^p \leq p a^{p-1}  b, &\text{for } a \geq b \geq 0.\snote{follows from \ref{eq:bound1a}}
\end{align}
For $0 \leq p \le 1$
\begin{align}
\label{eq:bound2}
&(1-\varepsilon)^p \le 1 - p \varepsilon,&\text{for } 0 \le \varepsilon \le 1,\\
\label{eq:bound3}
&(1-\varepsilon)^p \ge 1 - 2 p \varepsilon \ln 2,&\text{for }  0 \le \varepsilon \le 1/2,\\
\label{eq:bound3a}
&(1+\varepsilon)^p \ge 1 + p \varepsilon \ln 2,&\text{for } 0 \le \varepsilon \le 1,\\
\label{eq:bound_derived_p01}
&a^p-(a-b)^p \leq 2 p a^{p-1}  b \ln 2, &\text{for } a \geq 2b \geq 0.\snote{follows from \ref{eq:bound3}}
\end{align}

\subsection{Algorithm for \texorpdfstring{$p \ge 1$}{p ≥ 1}}
In this section we prove Theorem~\ref{th:large_p}. We start by constructing a family of functions $F_i$, which are better refinements of $F$ as $i$ decreases.

\paragraph{First step:}
Let us denote $$F_i(x,y) = \Big(\max(0, |x-y|-2^i)\Big)^p\quad\quad\text{and}\quad\quad f_i = F_i - F_{i+1}.$$
Observe that $F_{u} = 0$ (for $0\le x,y \le U$). Moreover, there is a telescoping sum $F_{i} = \sum\limits_{j = i}^{u} f_j$. To better see the the telescopic sum, consider case $p=1$. We then represent $F_{-u}(x,y) = \sum_{i=-u}^{u} f_i(x,y) = (- 2^{-u} + 2^{-u+1}) + (- 2^{-u+1} + 2^{-u+2}) + \ldots + (-2^{t-1} + 2^{t}) + (|x-y|-2^{t}) + 0 + \ldots + 0$. Such decomposition (for $p=1$) was first  considered, to our knowledge, in \cite{DBLP:journals/algorithmica/LipskyP11}.

\paragraph{Second step:}
Instead of using $x$ and $y$ for evaluation of $F_i$, we  evaluate $F_i$ using $x$ and $y$ with all bits younger than $i$-th one set to zero. Formally, define $x^{(i)} = \lfloor x/2^i \rfloor \cdot 2^i$, $y^{(i)} = \lfloor y/2^i \rfloor \cdot 2^i$. Now we denote
$$G_i(x,y) = F_i(x^{(i)},y^{(i)})$$

   Similarly as for $f_i$, define  $g_i = G_i - G_{i+1}$. Using the same reasoning, we have $G_u= 0$. For integers $i\leq 0$ the functions $F_i$ and $G_i$ are the same (as we are not rounding) and therefore $F_{-u} = G_{-u} = \sum\limits_{i = -u}^{u} g_i$. Intuitively, $g_i$ captures contribution of $i$-th bit of input to the output value (assuming all older bits are set and known, and all younger bits are unknown).

 \paragraph{Third step:}
 Let $\eta$ be a value to be fixed later, depending on $\varepsilon$ and $p$.
Assume w.l.o.g. that $\eta$ is such that $1/\eta$ is an integer.  We now define $\widehat{g}_i$ as a refinement of $g_i$, by replacing $|x^{(i)}-y^{(i)}|$ with $\|x^{(i)}-y^{(i)}\|_{B_i}$ and $|x^{(i+1)}-y^{(i+1)}|$ with $\|x^{(i+1)}-y^{(i+1)}\|_{B_i}$, where $B_i = 2^{i}/ \eta$, that is doing all the computation \emph{modulo} $B_i$. To be precise, define
\begin{align*}
\overrightarrow{G}_i(x,y) = & \Big(\max(0,\|x^{(i)}-y^{(i)}\|_{B_i}-2^i)\Big)^p\\
\overleftarrow{G}_{i+1}(x,y) = & \Big(\max(0,\|x^{(i+1)}-y^{(i+1)}\|_{B_i}-(2^{i+1})\Big)^p
\end{align*}
and then $\widehat{g}_i = \overrightarrow{G}_i - \overleftarrow{G}_{i+1}$. Additionally, we denote for short $\widehat{G}_i = \sum\limits_{j = i}^u \widehat{g}_j$.

Intuitively, $\widehat{g}_i$ approximates $g_i$ in the scenario of limited knowledge -- it estimates contribution of $i$-th bit of input to the output, assuming knowledge of bits $i+1$ to $i+\log \eta^{-1}$ of input.
We are now ready to provide an approximation algorithm to $(\ell_p)^p$ text-to-pattern distances.

\begin{algo}
\label{algo1}
\
\\ \textbf{Input:}
\begin{itemize}
    \item $T$ is the text,
    \item $P$ is the pattern,
    \item $\eta$ controls the precision of the approximation.
\end{itemize}
\textbf{Steps:}
\begin{enumerate}
\item For each $i \in \{-u,\ldots,u\}$ compute array $S_i$ being the text-to-pattern distance between $T$ and $P$ using $\widehat{g}_i$ distance function (parametrized by $\eta$) using Theorem~\ref{th:lipskyporat}.
\item Output array $S_{\varepsilon}[i] = \left(\sum\limits_{j=-u}^u S_j[i]\right)^{1/p}$.
\end{enumerate}
\end{algo}
\emph{To get the $(1+\varepsilon)$ approximation we run the Algorithm \ref{algo1} with $\eta = \frac{\varepsilon}{128}$.}
\bigskip

Now, we need to show the running time and correctness of the result. Firstly, to prove the correctness, we divide summands $\widehat{g}_i$ into three groups and reason about them separately. As computing  $F_{-u} $, $G_{-u} $(by summing $f_i$'s and $g_i$'s respectively) yields $(1+\varepsilon)$ multiplicative error, we will show that the difference between computing $g_i$ and $\widehat{g}_i$ brings only an additional $(1+\varepsilon)$ multiplicative error.
\begin{lemma}
\label{lem:l1}
For $i$ such that $|x-y| \le 2^i $ both $g_i(x,y) = 0$ and $\widehat{g}_i(x,y)=0$.
\end{lemma}
\begin{proof}
As both $g_i,\widehat{g}_i$ are symmetric functions, we can w.l.o.g. assume $x \geq y$. $\forall j \geq i$:

$$\left|x^{(j)} - y^{(j)}\right| = 2^j\left(\floor*{\frac{x}{2^j} } -\floor*{ \frac{y}{2^j} }\right) \leq 2^j\left(\floor*{\frac{x}{2^j} } -\floor*{\frac{x-2^i}{2^j} }\right) \leq 2^j. $$

Therefore $G_j=0$ from which $g_i(x,y) = 0$ follows. And because $\|x^{(j)}-y^{(j)}\|_{B_j} \le |x^{(j)} - y^{(j)}|$ we have $\widehat{g}_i(x,y) = 0$ as well.
\end{proof}

\begin{lemma}
\label{lem:l2}
For $i$ such that $|x-y| > 2^i \ge 4\eta |x-y|$ we have $g_i(x,y) = \widehat{g}_i(x,y)$.
\end{lemma}
\begin{proof}
  For $g_i(x,y) = \widehat{g}_i(x,y)$ to hold, it is enough to show that both \emph{norms} $|\cdot|$ and $\|\cdot\|_{B_i}$ are the same for $x^{(i)}-y^{(i)}$ and $x^{(i+1)}-y^{(i+1)}$. This happens if the absolute values of the respective inputs are smaller than $B_i/2$. Let us bound both $|x^{(i)}-y^{(i)}|$ and $|x^{(i+1)} - y^{(i+1)}|$:
$$\max(|x^{(i)}-y^{(i)}|,|x^{(i+1)} - y^{(i+1)}|) \le |x-y| + 2^{i+1}\le 2^{i+1} (1 + \frac{1}{8\eta}).$$

We can w.l.o.g. assume $\eta \le 1/8$ in order to make $\frac{1}{8\eta}$ a dominant term in the parentheses and reach:
$$\max(|x^{(i)}-y^{(i)}|,|x^{(i+1)} - y^{(i+1)}|) \le 2^{i+1} (1 + \frac{1}{8\eta}) \le \frac{2^{i}}{2\eta } = \frac{B_i}{2}. $$

Therefore $\|x^{(i)}-y^{(i)}\|_{B_i} = |x^{(i)} - y^{(i)}|$ as well as $\|x^{(i+1)}-y^{(i+1)}\|_{B_i} = |x^{(i+1)} - y^{(i+1)}|$ which completes the proof.
\end{proof}
\begin{lemma}
\label{lem:l3}
If $p \ge 1$ then for $i$ such that $4 \eta |x-y| > 2^i$ we have $|g_i(x,y)| \le 2p 2^i \cdot |x-y|^{p-1}$.
\end{lemma}
\begin{proof}
For the sake of the proof, we will w.l.o.g. assume $\eta \le 1/8$. Denote $A = |x^{(i)} - y^{(i)}|$, $B = |x^{(i+1)} - y^{(i+1)}|$, $A' = \max(0,A-2^i)$ and $B' = \max(0,B-2^{i+1})$. Observe that $|x-y| - 2^i \le A \le |x-y|+2^i$ thus $|x-y|-2\cdot 2^{i} \le A' \le |x-y|$, and similarly $|x-y|-2\cdot 2^{i+1} \le B' \le |x-y|$ so $|A'-B'| \le 2 \cdot 2^i$. Assume w.l.o.g. that $A' \ge B'$. We bound
\begin{flushright}\begin{minipage}{0.75\textwidth}
\begin{flalign*}|g_i(x,y)| &= (A')^p - (A' - (A'-B'))^p\\
&\le p (A'-B')  (A')^{p-1} \snote{by \eqref{eq:bound_derived_p1inf}}\\
&\le 2p 2^i \cdot |x-y|^{p-1}&&\qedhere
\end{flalign*}
\end{minipage}
\end{flushright}
\end{proof}
\begin{lemma}
\label{lem:l4}
If $p \ge 1$ then for $i$ such that $4 \eta |x-y| > 2^i$ we have $|\widehat{g}_i(x,y)| \le 2p 2^i \cdot |x-y|^{p-1}$.
\end{lemma}
\begin{proof}
Follows by the same proof strategy as in proof of Lemma~\ref{lem:l3}, replacing $|\cdot|$ with $\|\cdot\|_{B_i}$.
\end{proof}

\begin{theorem}
\label{th:approx_exact}

$\widehat{G}_{-u} = \sum\limits_{i \ge -u} \widehat{g}_i$ approximates $F_{-u}$ up to an additive $32 \cdot p \cdot \eta \cdot  |x-y|^p$ term.
\end{theorem}
\begin{proof}
We bound the difference between two terms:
\begin{align*}
|F_{-u}(x,y) - \sum_{i=-u}^{u} \widehat{g}_i(x,y)| &\le \sum_{i = -u}^{\log_2 (4\eta |x-y|)} \left(|\widehat{g}_i(x,y)|+|g_i(x,y)|\right)\\
&\le 2\cdot \left(\sum_{i = -\infty}^{\log_2 (4\eta |x-y|)} 2^i \right) \cdot 2 \cdot p \cdot |x-y|^{p-1}\\
&\le 32 \cdot \eta|x-y| \cdot p \cdot |x-y|^{p-1}
\end{align*}
where the bound follows from Lemma \ref{lem:l1}, \ref{lem:l2}, \ref{lem:l3} and \ref{lem:l4}.
\end{proof}

We now show that $F_{-u}$ is a close approximation of $D$ (recall $D(x,y) = |x-y|^p$).
\begin{lemma}
\label{F_-u bound}
For integers $x,y$ there is $D(x,y) \cdot (1-(2 \ln 2)p / U) \le F_{-u}(x,y) \le D(x,y)$.
\end{lemma}
\begin{proof}
For $x=y$ the lemma trivially holds, so for the rest of the proof we will assume $x \neq y$.
As $x,y$ are integers only, their smallest non-zero distance is 1. As $-u<0$ the $|x-y|-2^{-u} > 0$ and we bound $|x-y|\cdot (1-1/U) \le \max(0,|x-y|-2^{-u}) \le |x-y|$. By \eqref{eq:bound1} (when $p \ge 1$) or \eqref{eq:bound3} (when $p \le 1$) the claim follows.
\end{proof}

 By combining Theorem~\ref{th:approx_exact} with the Lemma \ref{F_-u bound} above we conclude that additive error of Algorithm~\ref{algo1} at each position is $(32 p \cdot \eta+\frac{p}{U}) \cdot |x-y|^p = p (\varepsilon/4+1/U) \cdot  |x-y|^p \le p \varepsilon |x-y|^p$ (since w.l.o.g. $\varepsilon \ge 4/U$), thus the relative error is $(1+p \varepsilon/2)$.

Observe that each $\widehat{g}_i$ is effectively a function over the alphabet of size $B_i / 2^i = 1/\eta$. Thus, the complexity of computing text-to-pattern distance using $\widehat{g}_i$ distance is $\bigo(\eta^{-1} n \log m)$, and iterating over at most $2u$ summands makes the total time $\bigo(\varepsilon^{-1} n \log m \log U)$.

Finally, since $p \ge 1$ and w.l.o.g. $\varepsilon \le 1/p$, by \eqref{eq:bound1a} and \eqref{eq:bound4a} $(1+p\varepsilon/2)$ approximation of $\ell_p^p$ distances is enough to guarantee $(1+\varepsilon)$ approximation of $\ell_p$ distances.

\subsection{Algorithm for \texorpdfstring{$0 < p \le 1$}{0 < p ≤ 1}}
\label{sec:smallp}

In this section we prove Theorem~\ref{th:small_p}. We note that the algorithm presented in the previous section does not work, since in the proof of Lemma \ref{lem:l3} and \ref{lem:l4} we used the convexity of function $|t|^p$, which is no longer the case when $p<1$.

However, we observe that Lemma \ref{lem:l1} and \ref{lem:l2} hold even when $0 < p \le 1$. To combat the situation where adversarial input makes the estimates in Lemma \ref{lem:l3} and \ref{lem:l4} to grow too large, we use a very weak version of hashing. Specifically, we pick at random a linear function $\sigma(t) = r \cdot t $, where $r \in [1,9)$ is a random independent variable. Such function applied to the input makes its bit sequences appear more ''random'' while preserving the inner structure of the problem.

Consider a following approach:

\begin{algo}
\label{algo2}
\
\begin{enumerate}
\item Fix $\eta = \frac{\varepsilon \cdot p}{\etconst \log U \ln 2}$.
\item Pick $r \in [1,9)$ uniformly at random.
\item Compute $T' = r \cdot T$ and $P' = r \cdot P$.
\item Use Algorithm \ref{algo1} to compute $S'$, text-to-pattern distance between $T'$ and $P'$ using $\widehat{G}_{-u}$ distance function.
\item Output $S'' = S' \cdot  r^{-1}$.
\end{enumerate}
\end{algo}
\bigskip

Now we analyze the expected error made by estimation from Algorithm~\ref{algo2}. We denote the expected additive error of estimation of $(\ell_p)^p$ distances as $$\text{err}(x,y) \defeq \ExpR{[\left(\frac{1}{r}\right)^p\left|\widehat{G}_{-u}(rx,ry) - |rx-ry|^p\right|}.$$

\begin{theorem}
The procedure of Algorithm~\ref{algo2} has the \textbf{expected} additive error $\text{err}(x,y) \le \frac{\varepsilon p}{3 \ln 2} |x-y|^p$.
\end{theorem}
\begin{proof}
Assume that $x \not= y$, as otherwise the bound trivially follows.
We bound the absolute error as follow, denoting $k = \log(8 \eta |x-y|)$).

\begin{align*}
\text{err}(x,y)&\le \ExpR{\left(\frac{1}{r}\right)^p \left|\widehat{G}_{-u}(r x,r  y) - F_{-u}(r x,r y) \right|}\\
&\hspace{2em}+\ExpR{\left|F_{-u}(rx,ry) - D(rx,ry)\right|}\\
&\le \ExpR{\left| \sum_{i=-u}^{u} \left(\widehat{g}_i(r x,r  y ) - g_i(r x,r  y )\right)\right|}\\ 
&\hspace{2em}+ \ExpR{\left(\frac{1}{r}\right)^p 2 (\ln 2)\frac{p}{U} D(rx,ry)}\snote{$(1/r)^p \le 1$}\\
&\le   \sum_{i=-u}^{k} \ExpR{\left|  \widehat{g}_i(r x , ry) \right|} + \ExpR{ \left|\sum_{i=-u}^{k} g_i(r x,ry )\right|}\\
&\hspace{2em}+ 2 (\ln 2) \frac{p}{U} |x-y|^p  \snote{Lemma \ref{lem:l1}, \ref{lem:l2}}
\end{align*}
Now, we bound the first two summands separately in following lemmas.
\begin{lemma}
$|\sum_{i=-u}^{k} g_i(r x,ry )|$ is upper bounded by $32 (\ln 2) \eta |x-y|^p$.
\label{summand_1}
\end{lemma}
\begin{proof}
Since w.l.o.g. $\eta \le 1/32$ thus $2^{k+1} \le 1/2 \cdot r|x-y|$):
\begin{flushright}\begin{minipage}{0.9\textwidth}\begin{align*}
\left|\sum_{i=-u}^{k} g_i(r x,ry )\right| &\le \left|\sum_{i=-\infty}^{k} g_i(rx,ry)\right|\\
&\le | G_{k+1}(rx,ry) - D(rx,ry) |\\
&\le ( (r|x-y|)^p - (r|x-y|-2^{k+1})^p)\\
&\le r^p |x-y|^p \cdot 2 p (\ln 2) \frac{2^{k+1}}{r|x-y|}&\snote{by \eqref{eq:bound_derived_p01}}\\
&\le 32 (\ln 2) \eta |x-y|^p. &\snote{$r^{p-1} \le 1$}\qedhere
\end{align*}
\end{minipage}\end{flushright}
\end{proof}

\begin{lemma}
\label{lem:ugly_lemma}
For $i \le k = \log(8 \eta |x-y|)$ we have $ \ExpR{\left|  \widehat{g}_i(r x , ry) \right|} \leq (\uglyconst) \eta |x-y|^p $.
\end{lemma}
\begin{proof}
First, we define symbols $A,B,A',B'$ to be parts of the $\widehat{g_i}$.
\begin{align*}
    A &= \|(rx)^{(i)} - (ry)^{(i)}\|_{B_i}\\
    B &=  \|(rx)^{(i+1)} - (ry)^{(i+1)}\|_{B_i}\\
    A' &= \max(0,A - 2^{i})\\
    B' &= \max(0,B-2^{i+1})
\end{align*}
Repeating reasoning from proof of Lemma~\ref{lem:l3}, we get
\begin{align}
|A' - B'| &\le 2 \cdot 2^i \nonumber\\
\|rx - ry\|_{B_i} - 2 \cdot 2^i &\le A' \le \|rx - ry\|_{B_i} \label{a_prime_bound}\\
\|rx - ry\|_{B_i} - 2 \cdot 2^{i+1} &\le B' \le \|rx - ry\|_{B_i} \label{b_prime_bound}
\end{align}
We also bound $B_i = 2^i/\eta \le 2^k / \eta = 8|x-y|$. Now let's bound the $|\widehat{g}_i(r x , ry)|$. A simple bound that comes from the definition of $\widehat{g}_i$ gives us:
\begin{align}
    |\widehat{g}_i(r x , ry)| = |A'^p - B'^p| \leq \max(A'^p,B'^p) \leq \|rx - ry\|_{B_i}^p.\snote{Use of \ref{a_prime_bound},\ref{b_prime_bound}}
    \label{g_i_simple_bound}
\end{align}
Unfortunately, this bound is not tight enough for larger values of $\|rx - ry\|_{B_i}$, so for $\|rx - ry\|_{B_i} \ge 6 \cdot 2^i$, we prove stronger bound:
\begin{align*}
|\widehat{g}_i(r x , ry)| &= |(A')^p - (B')^p|\\
&= \max(A',B')^p - \min(A',B')^p\\
&= \max(A',B')^p - (\max(A',B') - |A'-B'|)^p\\
&= \max(A',B')^p  \left(1 - \left(1 - \frac{|A'-B'|}{\max(A',B')}\right)^p\right)\\
&\le\|rx - ry\|_{B_i}^p \cdot (1 - (1-\frac{2\cdot 2^i}{\|rx - ry\|_{B_i}-2\cdot 2^i})^p)\\
&\le\|rx - ry\|_{B_i}^p \cdot (1 - (1-\frac{3\cdot 2^i}{\|rx - ry\|_{B_i}})^p)\\
&\le6 p (\ln 2)  \|rx - ry\|_{B_i}^{p-1} \cdot 2^i \bnote{$\|rx - ry\|_{B_i} \ge 6 \cdot 2^i$, by \eqref{eq:bound3}}.
\end{align*}

The \emph{norm} function $\| x \|_{B_i} = \min( x \bmod B_i, B_i - (x \bmod B_i))$ is in fact a triangle wave function varying between $0$ and $B_i/2$ with periodicity of $B_i$. So if the input is a random variable that follows uniform distribution at interval that is larger than its period (in our case $B_i$), the output has piece-wise uniform distribution, and its probability density function can be bounded by two times the probability density function of the uniform distribution for the whole domain. Formally, if $X = U(a,b)$ with $b-a \geq B_i$ then for $Y = \| X \|_{B_i}$ its probability density function $f_Y(y)$ is:
\begin{align}
    f_Y(y) \leq \frac{2}{B_i/2} \text{ for } 0 \leq y \leq B_i/2
    \label{density_bound}
\end{align}

As the input in the expression $\|rx - ry\|_{B_i}$ to the \emph{norm} function is uniformly distributed between $a = |x-y|$ and $b = 9 |x-y|$ and $B_i \leq 8 |x-y|$, we can use \ref{density_bound} to bound the probability density function of the $Z = \|rx - ry\|_{B_i}$ by $f_Z(y) \leq \frac{2}{B_i/2}$.

Now when we have the approximate probability density function (namely its upper bound) we can condition on the value of $\|rx-ry\|_{B_i}$ to be able to use the bounds for small and large values of $\|rx-ry\|_{B_i}$.
\begin{align*}
\ExpR{&|\widehat{g}_i(rx , ry)|} =\\
&= \ExpR{|\widehat{g}_i(rx , ry)|\ \given\ \|rx-ry\|_{B_i} \le 6 \eta B_i}\PrR{\|rx-ry\|_{B_i} \le 6 \eta B_i} +\\
&\qquad+\ExpR{|\widehat{g}_i(rx , ry)|\ \given\ \|rx-ry\|_{B_i} > 6 \eta B_i}\PrR{\|rx-ry\|_{B_i} > 6 \eta B_i}
\end{align*}
We bound those two summands separately. Now, bound on the first part:
\begin{align*}
\ExpR{&|\widehat{g}_i(rx , ry)|\ \given\ \|rx-ry\|_{B_i} \le 6 \eta B_i}\PrR{\|rx-ry\|_{B_i} \le 6 \eta B_i} \le \\
& \leq \ExpR{|\widehat{g}_i(rx , ry)|\ \given\ \|rx-ry\|_{B_i} \le 6 \eta B_i} 24\eta \\
& \leq \ExpR{\|rx-ry\|_{B_i}^p\ \given\ \|rx-ry\|_{B_i} \le 6 \eta B_i} 24\eta \snote{by \ref{g_i_simple_bound}}\\
& \leq (6 \eta  B_i)^p 24 \eta\\
& \leq 24 \eta (6 \cdot 2^i)^p\\
& \leq 24 \cdot 6 \eta (8\eta |x-y|)^p\\
& \leq 1152 \eta |x-y|^p
\end{align*}
And on the second part:
\begin{align*}
\ExpR{&|\widehat{g}_i(rx , ry)| \given \|rx-ry\|_{B_i} > 6 \eta B_i}\PrR{\|rx-ry\|_{B_i} > 6 \eta B_i} \le \\
& \leq \ExpR{ 6 p (\ln 2)  \|rx - ry\|_{B_i}^{p-1} \cdot 2^i \given \|rx-ry\|_{B_i} > 6 \eta B_i} \cdot\\
&\qquad \qquad \cdot \PrR{\|rx-ry\|_{B_i} > 6 \eta B_i} \\
& \leq \int_1^9 6 p (\ln 2)  \|rx - ry\|_{B_i}^{p-1} \cdot 2^i \frac{1}{8} \cdot \indicator{\|rx-ry\|_{B_i}>6 \eta B_i} dr\bnote{1/8 is the \\ density of r.v. $r$,\\$\indicator{\cdot}$ is the indicator \\ function}\\
& \leq 6 p (\ln 2) \cdot 2^i \int_{0}^{B_i/2} z^{p-1} \frac{2}{B_i/2} \indicator{z>6 \eta B_i} dz \bnote{changed to r.v. \\ $z=\|rx-ry\|_{B_i}$}\\
& \leq p (\ln 2) \frac{24}{B_i}\cdot 2^i \int_{0}^{B_i/2} z^{p-1} dz\\
& \leq (\ln 2) \frac{24}{B_i}\cdot 2^i \left(\frac{B_i}{2}\right)^{p}\\
& \leq 24 (\ln 2)  B_i^{p-1}\cdot 2^i\\
& \leq 24 (\ln 2) 2^{ip} \eta^{-p+1}\\
& \leq 24 (\ln 2) (8\eta|x-y|)^p \eta^{-p+1}\\
& \leq 192 (\ln 2) \eta|x-y|^p
\end{align*}
So finally, we reach:
\begin{align*}
\ExpR{&|\widehat{g}_i(rx , ry)|} =\\
&= \ExpR{|\widehat{g}_i(rx , ry)| \given \|rx-ry\|_{B_i} \le 6 \eta B_i}\PrR{\|rx-ry\|_{B_i} \le 6 \eta B_i} +\\
&\qquad+\ExpR{|\widehat{g}_i(rx , ry)| \given \|rx-ry\|_{B_i} > 6 \eta B_i}\PrR{\|rx-ry\|_{B_i} > 6 \eta B_i}\\
&\leq 1152 \eta |x-y|^p + 192 (\ln 2) \eta|x-y|^p \\
&\leq (\uglyconst) \eta |x-y|^p&&\qedhere
\end{align*}
\end{proof}
By combining bounds from Lemma \ref{summand_1}, and Lemma \ref{lem:ugly_lemma} we get:
\begin{align*}
\textrm{err}(x,y)&\le   \sum_{i=-u}^{k} \ExpR{ \left|  \widehat{g}_i(r x , ry) \right|} +  \ExpR{\left|\sum_{i=-u}^{k} g_i(r x,ry )\right|} + 2 (\ln 2) \frac{p}{U} |x-y|^p\\
 &\le 2 (\ln 2) \frac{p}{U}  |x-y|^p  + (32 (\ln 2) \eta |x-y|^p + \sum_{i=-u}^u (\uglyconst) \eta |x-y|^p\\
 &\le 2 (\ln 2) \frac{p}{U}  |x-y|^p  + (32 (\ln 2) \eta |x-y|^p + 2 \log U (\uglyconst) \eta |x-y|^p\\
 &\le 2 (\ln 2) \frac{p}{U}  |x-y|^p + (32 (\ln 2) + 2 \log U (\uglyconst)) \eta |x-y|^p\\
&\le 2 (\ln 2) \frac{p}{U}  |x-y|^p+ \mergedconst \log U \eta |x-y|^p\\
&\le \frac{\varepsilon p}{6 \ln 2} |x-y|^p + \frac{\varepsilon p}{6 \ln 2} |x-y|^p&  \text{w.l.o.g.}\  \varepsilon \ge  \frac{12 (\ln 2)^2}{U}\\
&\le \frac{\varepsilon p}{3 \ln 2} |x-y|^p&&\qedhere
\end{align*}
\end{proof}

To finish the proof of Theorem~\ref{th:small_p} we observe, that for any position $i$ of output, Algorithm~\ref{algo2} outputs $S''[i]$ such that $\mathbb{E}_r[|(S''[i])^p-(S[i])^p|] \le \frac{p \varepsilon}{3 \ln 2} \cdot (S[i])^p$. By Markov's inequality it means that with probability $2/3$ the relative error of $(\ell_p)^p$ approximation is at most $\frac{p}{\ln 2} \cdot \varepsilon$. Thus, by \eqref{eq:bound2} and \eqref{eq:bound3a} relative error of $\ell_p$ approximation is $\varepsilon$ with probability at least $2/3$. Now a standard amplification procedure follows: invoke Algorithm~\ref{algo2} independently $t$ times and take the median value from $S''_{(1)}[i], \ldots S''_{(t)}[i]$ as the final estimate $S_{\varepsilon}[i]$. Taking $t= \Theta(\log n)$ to be large enough  makes the final estimate good with
high probability, and by the union bound whole $S_{\varepsilon}$ is a good estimate of $S$. The complexity of the whole procedure is thus $\bigo(\log n \cdot \log U \cdot \eta^{-1} \cdot n \log m ) = \bigo(p^{-1} \varepsilon^{-1}n \log m \log^2 U \log n)$.

\section{Hamming distances}
As a final note we comment on a particularly simple form that Algorithm~\ref{algo2} takes for Hamming distances (limit case of $p=0$).

$$\widehat{g}_i(x,y) = \begin{cases}1 &\quad\text{if } \|x^{(i)}-y^{(i)}\|_{B_i} = 1\\0 &\quad\text{otherwise, } \end{cases}$$
with Algorithm being simply: pick at random $r \in [1,9]$, apply it multiplicatively to the input, compute text-to-pattern distance using $\sum_i \widehat{g}_i$ function.

Taking a limit of $p \to 0$ in proof of Theorem~\ref{th:small_p}, we reach that bound from Lemma~\ref{lem:ugly_lemma} becomes
$$\ExpR{|\widehat{g}_i(r x , ry)|} \le 24 \eta$$
and since all other terms in error estimate have multiplicative term $p$ in front, we reach
$$\textrm{err}(x,y) \le 2 \log U \cdot \ExpR{ |\widehat{g}_i(r x , ry)|} \le 48 \eta\log U.$$
We thus observe that expected relative error in estimation of Hamming distance is: $\mathbb{E}[S''[i] - S[i]] \le 48 \eta\log U \cdot S[i]$. With probability at least $2/3$ the relative error is at most $144 \eta \log U$. Setting $\eta = \frac{\varepsilon}{144 \log U}$ and repeating the randomized procedure $\Theta(\log n)$ with taking median for concentration completes the algorithm. The total runtime is, by a standard trick of reducing alphabet size to $2m$,  $\bigo(\frac{n}{\varepsilon} \log^2 m \log n)$, and while it compares unfavorably to algorithm from \cite{DBLP:conf/soda/KopelowitzP18} (in terms of runtime), it gives another insight on why $\widetilde\bigo(n/\varepsilon)$ time algorithm is possible for Hamming distance version of pattern matching.

\bibliography{bib}

\end{document}